\documentclass[journal]{IEEEtran}
\usepackage{epsfig,amsmath,amssymb,epsf,amsthm,scalefnt,multirow,subfig}
\usepackage{xcolor}
\usepackage{float}
\usepackage{cite}
\usepackage{psfrag}
\usepackage{array}
\usepackage{verbatim}
\usepackage{tikz}
\usetikzlibrary{tikzmark,calc,decorations.pathreplacing}
\usepackage{amsmath}

\newcommand{\argmax}{\mathop{\mathrm{argmax}}}

\newtheorem{theorem}{Theorem}

\newtheorem*{lemma*}{Lemma}
\newtheorem{remark}{Remark}

\newtheorem{proposition}{Proposition}

\def\argmax{\mathop{\mathrm{argmax}}}

 \def\bomega  
\def\bTheta{{\pmb{\Theta}}}

\def\b0{{\pmb{0}}} 

\allowdisplaybreaks[4]

\begin{document}

\title{ Anti-Jamming Games \\ in Multi-Band Wireless Ad Hoc Networks}


\author{Hyeon-Seong Im,~\IEEEmembership{Graduate Student Member,~IEEE,}
        and Si-Hyeon Lee,~\IEEEmembership{Senior Member,~IEEE}
\thanks{H.-S. Im and S.-H. Lee (Corresponding Author) are with the School of Electrical Engineering, Korea Advanced Institute of Science and Technology (KAIST), Daejeon, South Korea (e-mail: imhyun1209@kaist.ac.kr, sihyeon@kaist.ac.kr). A short version of this paper was submitted to IEEE GLOBECOM 2022 \cite{Im:2022_globecom}.}
}

\maketitle

\begin{abstract}
For multi-band wireless ad hoc networks of multiple users, an anti-jamming game between the users and a jammer is studied. In this game, the users (resp. jammer) want to maximize (resp. minimize) the expected rewards of the users taking into account various factors such as communication rate, hopping cost, and jamming loss. We analyze the arms race of the game and derive an optimal frequency hopping policy at each stage of the arms race based on the Markov decision process (MDP). It is analytically shown that the arms race reaches an equilibrium after a few rounds, and a frequency hopping policy and a jamming strategy at the equilibrium are characterized. We propose two kinds of collision avoidance protocols to ensure that at most one user communicates in each frequency band, and provide various numerical results that show the effects of the reward parameters and collision avoidance protocols on the optimal frequency hopping policy and the expected rewards at the equilibrium. Moreover, we discuss about  equilibria for the case where the jammer adopts some unpredictable jamming strategies. 
\end{abstract}  

\section{Introduction}\label{sec1}
In wireless ad hoc networks, multiple sender-receiver pairs communicate in a distributed manner without the help of infrastructures such as base stations \cite{Toh:1996,Rubinstein:2006,Sharmila:2016}. The applications of wireless ad hoc networks include military communications~\cite{Zhou:1999} and emergency communications \cite{Gyoda:2008} for which it is inherently hard to access infrastructures,  local networks of nearby devices like personal area networks or Internet of Things (IoT) networks \cite{Al-Fuqaha:2015}, and cognitive radio communications where secondary users communicate over empty channels \cite{Akyildiz:2009}. Many wireless ad hoc networks operate in strong interference regime \cite{Yaqoob:2017,Zain-Ul-Abideen:2019}, which makes them vulnerable to jamming attacks \cite{Xu:2005,Aristides:2009} as  jamming attacks with low power can lead to communication failures. Moreover,  jamming attacks are more fatal to mobile ad hoc networks of battery-operated devices, by making the devices retransmit the data packets repeatedly~\cite{Yaqoob:2017, Weniger:2004}. 

An effective anti-jamming strategy is to perform frequency hopping over multiple frequency bands (multiple channels) \cite{Zain-Ul-Abideen:2019, Khattab:2010,Khattab:2008,Navda:2007,Pickholtz:1982}. A simple frequency hopping policy is to hop to other channels only if a jamming attack is detected in the current channel to minimize the hopping cost \cite{Khattab:2010,Khattab:2008}. This reactive frequency hopping is effective for random jamming attacks that randomly choose the channels to attack. For more intelligent jamming attacks, a proactive frequency hopping policy can perform better where a user stays in a channel for a fixed time and then proactively hops to other channels \cite{Khattab:2008,Navda:2007}. This proactive frequency hopping policy is also appropriate when it is not easy for a user to detect a jamming attack. The classical frequency hopping spread spectrum   can be also used to mitigate jamming attacks without  jamming detection by rapidly changing the channels  \cite{Pickholtz:1982}. We note that in general, the performance accessment of a frequency hopping policy depends on the type of the jamming attack, and vice versa. 

By taking into account such an interplay between frequency hopping policies and jamming attacks, a game-theoretic approach has been taken for anti-jamming studies. The works \cite{Altman:2007,Garnaev:2019,Wu:2009,Guizani:2020} studied  power allocation games between users and jammers, under the assumption of full channel state information over all the frequency bands at the users and at the jammers. The goal of the game for the users (resp. jammers) is to maximize (resp. minimize) the sum  capacity over all channels  \cite{Altman:2007,Garnaev:2019}. The work \cite{Altman:2007} considered a single-user scenario with transmission cost and showed that a generalized water-filling achieves the unique Nash equilibrium. A similar setting with \cite{Altman:2007} was considered in \cite{Garnaev:2019}   in the presence of an eavesdropper.  Such power allocation problems have been also modeled as Colonel-Blotto games \cite{Roberson:2006} where users try to maximize the number of the channels whose signal-to-noise-and-interference ratios  exceed a certain threshold \cite{Wu:2009,Guizani:2020}. 

For more complicated scenarios where it is hard to directly analyze an equilibrium, the analysis of arms race can give good approximations to a game equilibrium. In the arms race, users and jammers alternately improve their strategies in a way that the users (resp. jammers) find the best frequency hopping policy (resp. jamming strategy) against the current jamming strategy (resp. frequency hopping policy). An optimal frequency hopping policy against a given jamming strategy can be derived based on the Markov decision process (MDP) \cite{Howard:1960}, which in general aims at finding an optimal policy that specifies an optimal action to be taken at each state that maximizes the expected reward. In the derivation of an optimal frequency hopping policy, an MDP-based approach can be useful because (i) it can utilize the information about the past experience, e.g., how long it stayed the current channel, whether it was jammed or collided with other users, etc., by defining the states appropriately, (ii) it can capture various cost functions such as hopping cost and  jamming loss by defining the reward accordingly,  and (iii) it can maximize the accumulated reward over time. The arms race in anti-jamming games has been studied mainly for single-user cases. For cognitive networks with a single secondary user \cite{Wu:2012}, the first few rounds of the arms race were analyzed and  a frequency hopping policy was derived based on the MDP. It was shown that the game reaches near an equilibrium after a few rounds of the arms race. The works \cite{Hanawal:2020,Hanawal:2016} considered the arms race when a user has some options of communication modes, i.e., the user can choose either in-band-full-duplex (IBFD) or half-duplex (HD) modes in  \cite{Hanawal:2020} and it can choose a communication rate among a finite set of candidates in \cite{Hanawal:2016}. 

In this paper, we consider an anti-jamming game in wireless ad hoc networks of multiple users. For multi-user networks, the interference among the users needs to be considered, which makes the analysis of anti-jamming games complicated. Due to its intractability, most of the previous works on anti-jamming games in multi-user scenarios considered the problem of finding an optimal  frequency hopping policy based on Q-learning \cite{Watkins:1992},  which approximates the environment in an empirical way, against unknown jamming attacks, instead of analyzing Nash  equilibria of the games. In particular, the work \cite{Yao:2019} considered the scenario where the users collaborate to maximize their global reward. It formulated the problem as a multi-agent MDP and proposed a Q-learning based algorithm which learns optimal frequency hopping policies of the users to maximize the expected global reward. The works \cite{Aref:2017,Shan:2018} considered a setting where each user is equipped with a sensing part and a transmission part. In \cite{Aref:2017}, the sensing part learns the strategies of the jammer and the other users, and the transmission part derives an optimal frequency hopping policy based on the sensing result, where both the parts use Q-learning algorithms for learning. In these works, the performance of Q-learning based algorithms was numerically analyzed under some jamming strategies. However, the analysis of the arms race and the characterization of equilibria for anti-jamming games in multi-user networks have not been well studied. 

Our main contribution in this paper is in analyzing an equilibrium of an anti-jamming game in wireless ad hoc networks. We consider a multi-band wireless ad hoc network where multiple users want to communicate in a fully distributed manner in the presence of a jammer. Each user aims at finding a frequency hopping policy that maximizes its expected reward, which is a function of communication throughput, hopping cost, and jamming loss, and at the same time, the jammer targets to minimize the users' rewards by effectively choosing the channels to attack. For this anti-jamming game, we analyze the arms race and derive an optimal frequency hopping policy  at each stage of the arms race based on the MDP. 
The contribution of this paper can be summarized as follows: 
\begin{itemize}
\item For the multi-band wireless ad hoc network of multiple users, the arms race is shown to reach an equilibrium after a few rounds,  and a frequency hopping policy and a jamming strategy at the equilibrium are characterized. Compared to the single-user case, the main difficulty in characterizing an optimal frequency hopping policy based on the MDP for the multi-user case comes from the mutual effects among the users, i.e.,  the environment with respect to each user needs to be specified to derive an optimal frequency hopping policy based on  MDP, but the environment is affected by the frequency hopping policies of the other users. To tackle this, we analyze iterative updates between the environment and the users' policies and show the convergence in finite number of iterations.  

\item Various numerical results are presented to show the effects of the reward parameters and collision avoidance protocols on the frequency hopping policy and the expected rewards of the users  at the equilibrium. 

\item The jamming strategy at the equilibrium is shown to scan the channels successively according to a fixed sweeping pattern. Here the underlying assumption is that the users do not know the sweeping pattern, but it is possible that the users eventually learn the scanning pattern after a sufficiently long time. Accordingly, we consider  the  case where the jammer adopts some unpredictable jamming strategies and show the existence of an unpredictable jamming strategy at an equilibrium under some conditions. 
\end{itemize}

The outline of the remaining paper is as follows. We formally present the problem statement in Section \ref{sec2}, and formulate the MDP  for finding an optimal frequency hopping policy at each stage of the arms race in Section \ref{sec3}. A frequency hopping policy and a jamming strategy at an equilibrium are characterized by analyzing the arms race for the single-user case and for the multi-user case in Sections~\ref{sec4}~and~\ref{sec5}, respectively. The arms race for the case where the jammer adopts some unpredictable jamming strategies is analyzed in Section \ref{sec6}. Finally, we conclude this paper in  Section \ref{sec7}.

\section{Problem Statement}\label{sec2}
\subsection{User model}\label{sec2A}
Consider a time-slotted multi-band wireless ad hoc network with $n$ sender-receiver pairs  (S-R pairs or users)\footnote{In wireless ad hoc networks, the synchronization of the users can be achieved by the time synchronization function in IEEE 802.11 standards~\cite{Basagni:2004}.}.
The frequency band is divided into $M$ non-overlapping channels $\{c_1,c_2,...,c_M\}$. In each time slot, each sender tries to transmit its data packets to the intended receiver through one of the $M$ channels.  The communication delays are assumed to be negligible. 
The overall transmission model in each time slot consists of the following phases, as also illustrated in Fig. \ref{Fig1}.

\begin{figure}
\centering
\includegraphics[width=0.9\columnwidth]{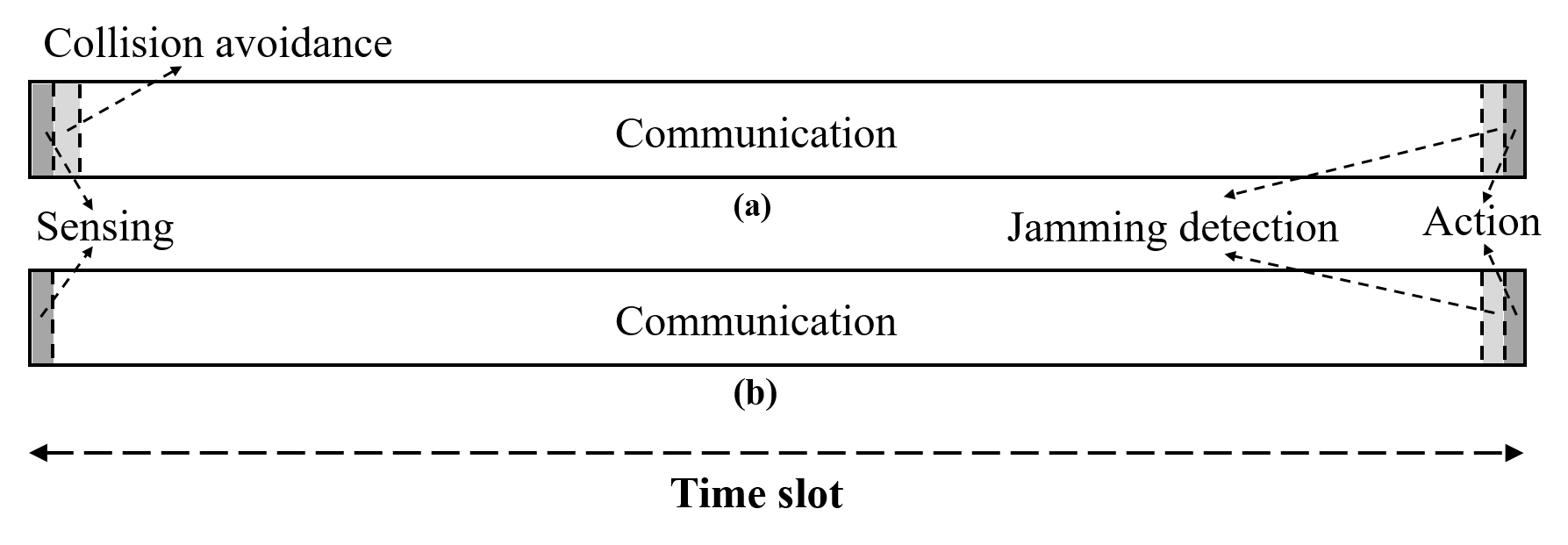}
\caption{Transmission protocol. $(a)$ and $(b)$ represent the users' protocol corresponding to the cases where the sensing phases are silent and non-silent, respectively.}\label{Fig1}
\end{figure}

\begin{itemize}
\item Sensing phase: The S-R pairs are assumed to use an error-correcting code designed for point-to-point communications, hence an appropriate collision avoidance protocol needs to be incorporated. We give priority to the S-R pair that already occupied the channel. To do so, at the beginning of each time slot, called sensing phase, the sender who already occupied the channel broadcasts a pilot signal. Then the S-R pairs who just hopped to the channel give up the communication and randomly hop to other channels in the next time slot. 

\item Collision avoidance phase: If there is an S-R pair who already occupied the channel, this phase is skipped. If not, i.e., it is silent in the sensing phase, since there can be multiple S-R pairs who newly hop to the same channel, we apply some collision avoidance protocols to ensure that at most one S-R pair communicates in the channel. We consider some symmetric collision avoidance protocols among the S-R pairs, which are described in Section~\ref{sec5}. 

\item Communication phase: The communication of the S-R pair, who already occupied the channel or who was allowed to communicate through the collision avoidance protocol, takes place. 

\item Jamming detection phase: The S-R pair, who  just finished the communication, judges whether the communication has been jammed. We assume that the jamming hypothesis test is correct with high probability.

\item Action phase: At the end of the time slot, each S-R pair determines whether to hop or not based on the history. This decision rule is called a frequency hopping policy.  Since we consider symmetric collision avoidance protocols among the S-R pairs, the frequency hopping policies of the S-R pairs are also assumed to be symmetric, i.e., the S-R pairs use the same policy.
\end{itemize}

If an S-R pair determines to hop, the next channel is selected uniformly at random among all $M$ channels.
Also, each S-R pair is assumed to share a sufficiently long pesudo-random sequence, so that they can hop to the same channel.  

\subsection{Jammer model}\label{sec2B}
A jammer scans $m<M$ channels in each time slot to test whether the S-R pairs are communicating or not. We assume that $m$ divides $M$, and let $T:=M/m$. If the jammer detects some communications in the scanned channels, then it attacks them by transmitting sufficiently large Gaussian noise so that the receivers cannot decode the data packets with high probability. We assume that the jamming attacks occur in the communication phase, i.e., the sensing and collision avoidance phases are not affected by the jamming attack. An example of communication scenario is illustrated in Fig. \ref{Fig2}. 

\begin{figure}
\centering
\includegraphics[width=0.9\columnwidth]{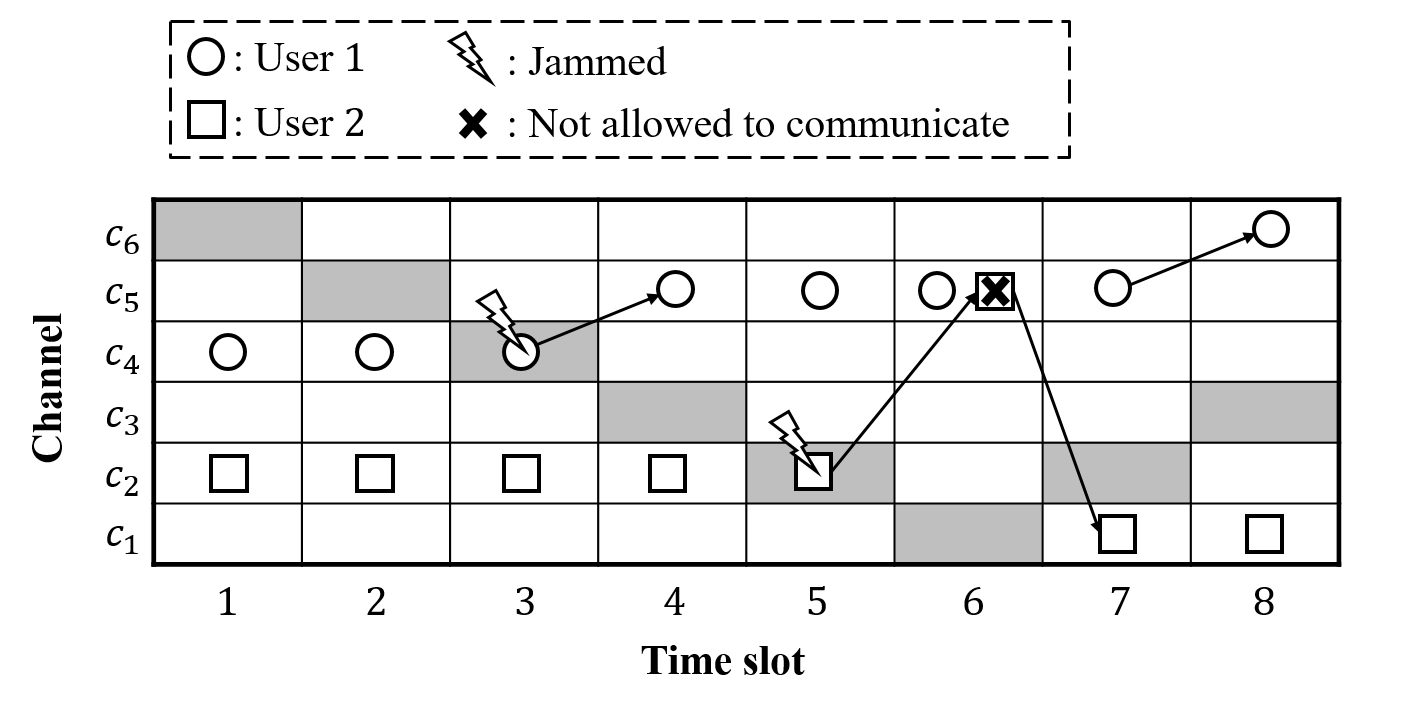}
\caption{An example of communication scenario, where $M=6$, $m=1$, and $n=2$. 
The channels scanned by the jammer are represented with shaded color. In this example, users 1 and 2 are jammed in time slots 3 and 5, respectively. In time slot 6, user 2 is not allowed to communicate in the communication phase since user 1 has already occupied the same channel. }\label{Fig2}
\end{figure}

\subsection{Arms race}\label{sec2C}
The S-R pairs and the jammer play a zero-sum game, i.e., the S-R pairs and the jammer try to maximize and minimize a ``reward", respectively. This reward, precisely defined in Section \ref{sec3}, is set to take into account  the communication throughput, jamming loss, hopping cost, and the priority of the current  compared to the future.  The zero-sum game can be described as an arms race  \cite{Wu:2012}. In the arms race, when the jammer changes the jamming strategy to minimize the reward, the S-R pairs change the frequency hopping policy to maximize it, and vice versa.  In some cases, as the arms race continues, the game reaches an  equilibrium, where the jamming strategy and the frequency hopping policy no longer change. In this paper, we analyze how the arms race proceeds and show pairs of frequency hopping policy and jamming strategy at an equilibrium.  

In TABLEs \ref{table1} and \ref{table2}, the jamming strategies and the frequency hopping policies that will appear in the arms race are summarized. Two kinds of jamming strategies appear in the arms race, i.e., $G$-memory jamming and reactive sweep jamming. The $G$-memory jamming for  $G\in[0:T-1]$ selects $m$ target channels uniformly at random out of all the $M$ channels except the channels scanned in the recent $G$ time slots. The special cases of $G=0$ and $G=T-1$ are called random jamming and basic sweep jamming \cite{Wu:2012}, respectively.  In contrast to that the $G$-memory jamming is not affected by whether the jammer detects a communication or not, the reactive sweep jammer performs the basic sweep jamming as long as it does not detect a communication, but once it detects a communication, it newly starts the basic sweep jamming with a randomly updated pattern.

\begin{table}
\centering
\begin{tabular}{|m{1.3cm}|| m{1.4cm}|m{4.3cm}|}
\hline
{Jamming strategy} & \multicolumn{2}{c|}{Definition}\\
\hline
\multirow{3}{1.5cm}{ } & \multicolumn{2}{m{6.3cm}|}{Select $m$ target channels uniformly at random out of all the $M$ channels except the channels scanned in the recent $G$ time slots.} \\
\cline{2-3}
$G$-memory jamming, for $G\in[0:T-1]$& $G=0$:  Random jamming & Select $m$ target channels uniformly at random out of all the $M$ channels for each time slot.\\
 \cline{2-3}
& $G= T-1$: Basic sweep jamming &  Scan all the $M$ channels in any window of $T$ consecutive time slots where the pattern of choosing $m$ channels across the time (sweep pattern) is fixed as shown in Fig. \ref{Fig3}.\\
\hline
 \multirow{1}{1.5cm}{\\ Reactive sweep jamming}  &\multicolumn{2}{p{6.3cm}|}{Perform the basic sweep jamming as long as the jammer does not detect a communication. Whenever the jammer detects a communication, it newly starts basic sweeping jamming with a randomly updated pattern.} \\
\hline
\end{tabular}
\caption{Jamming strategies}\label{table1}
\end{table}

\begin{table}
\centering
\begin{tabular}{|m{2cm}|| m{1.3cm} | m{3.5cm}}
\hline
Frequency hopping policy & \multicolumn{2}{c|}{Definition}\\
\hline
\multirow{2}{2.2cm}{\\ $K$-staying policy, for $K\in[0:\infty]$}  & \multicolumn{2}{m{5.6cm}|}{If the S-R pair is not allowed to communicate in the collision avoidance phase or it is jammed, it randomly and uniformly hop in the next time slot. If not, i.e.,  the S-R pair is allowed to communicate in the collision avoidance phase and it is not jammed, it stays the same channel if it has stayed there for $\leq K$ consecutive time slots, and it randomly hops in the next time slot if it has stayed there for $>K$ consecutive time slots. } \\
 \cline{2-3}
&$K=\infty$: Minimal hopping policy & \multicolumn{1}{m{3.8cm}|}{Stay at the same channel unless the S-R pair is jammed or it is not allowed to communicate in the collision avoidance phase. } \\
\hline
\end{tabular}
\caption{Frequency hopping policies}\label{table2}
\end{table}

\begin{figure}
\centering
\includegraphics[width=0.8\columnwidth]{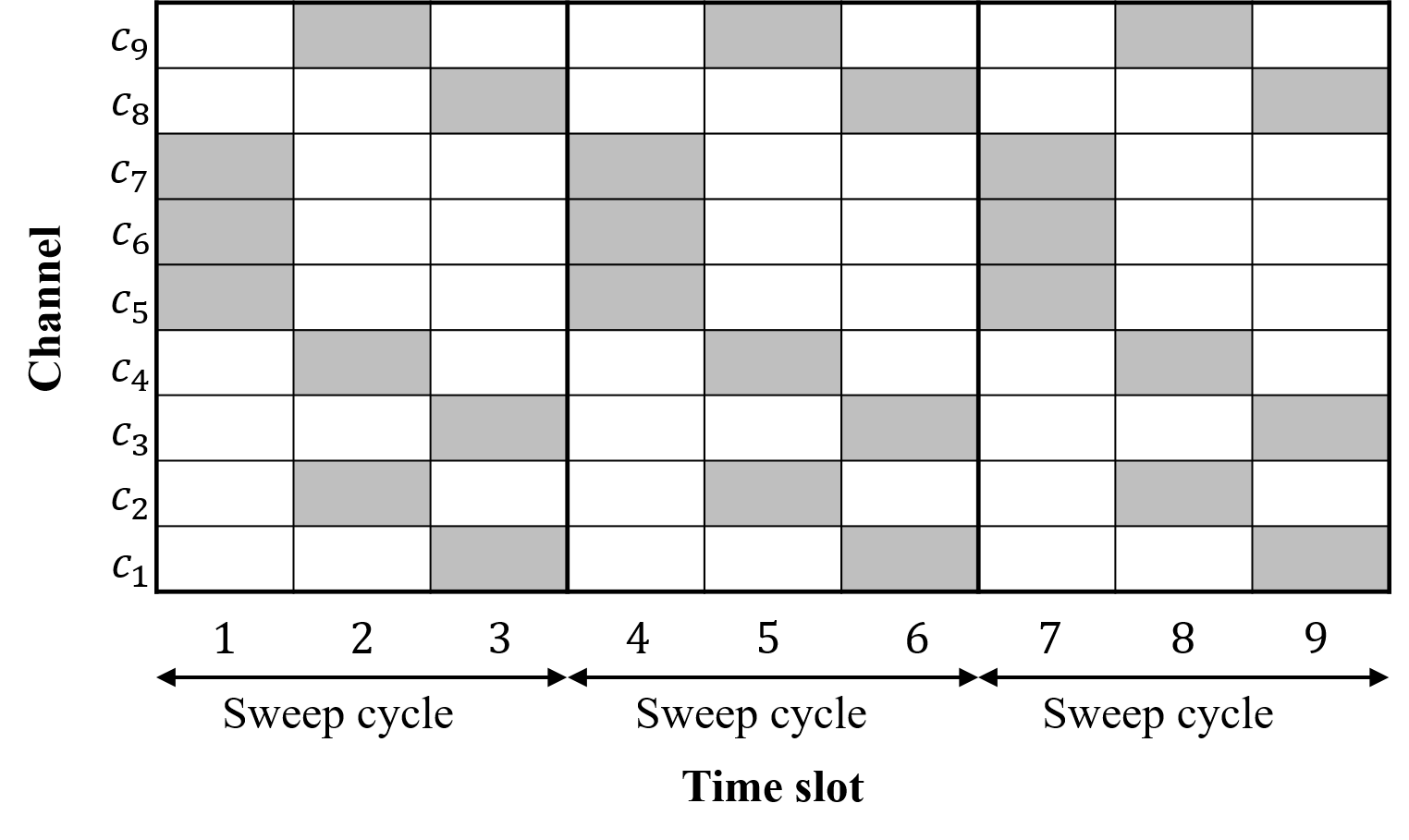}
\caption{An example of basic sweep jamming for $M=9$ and $m=3$. The channels scanned by the jammer are represented with shaded color.}\label{Fig3}
\end{figure}

For the frequency hopping policy,  we assume that an S-R pair always hops if an S-R pair is jammed or it is not allowed to communicate in the collision avoidance phase. In the arms race, the $K$-staying policy defined in Table \ref{table2} appears. In the $K$-staying policy, the S-R pair can stay at the same channel for at most $(K+1)$ consecutive time slots as long as it is allowed to communicate in the collision avoidance phase and it is not jammed. The special case of $K=\infty$ is called minimal hopping policy, where the S-R pair stays the same channel as long as it is allowed to communicate and it is not jammed.

%

The arms race starts with the naive random  jamming.  The combating hopping policy corresponding to the random jamming would be the minimal hopping policy. Since the random jamming newly and randomly chooses the set of $m$ channels for each time slot,  there is no reason for an S-R pair to hop by paying the hopping cost unless it is jammed or it is not allowed to communicate in the collision avoidance phase. In response to the minimal hopping policy of the S-R pairs, an optimal jamming strategy  is the basic sweep jamming \cite{Wu:2012}. Note that the basic sweep jammer scans all the $M$ channels using the least number of the time slots. In Sections \ref{sec4} and \ref{sec5}, we analyze the arms race that proceeds after the basic sweep jamming strategy and show that the next arms race against the basic sweep jammer is the $K$-staying policy for some $K$. Furthermore, it is shown that the basic sweep jamming and the $K$-staying policy with appropriately chosen $K$ form an equilibrium.

\section{MDP Formulation}\label{sec3}
For each S-R pair, the process of finding an optimal frequency hopping policy can be modeled as a Markov decision process (MDP) \cite{Howard:1960}, under the environment determined by the hopping policies of the other S-R pairs and the jamming strategy. 
The MDP model of an S-R pair is described by a tuple $(\mathcal{S},\mathcal{A},\mathcal{P},\mathcal{U},\gamma)$. First, the set $\mathcal{S}$ of states is given as $\{J,I,1,2,\cdots,K_m\}$. The state $S_t\in\mathcal{S}$ at time $t$  is determined at the end of the jamming detection phase. It becomes $K\in [1:K_m]$ if the S-R pair has successfully communicated for $K$ consecutive times up to time $t$, where $K_m$ denotes the largest possible number of consecutive times in which the communications of the S-R pair are  successful. This $K_m$ is affected by the jamming strategy and it can be infinite for some jamming strategies. The state $S_t$ becomes $I$ if the S-R pair did not get the chance to communicate at time $t$, and becomes $J$ if the S-R pair tried to communicate but it was jammed. The set $\mathcal{A}$ of actions that the S-R pair can take is $\{s,h\}$, i.e., stay in the same channel or hop uniformly at random. In the action phase of time $t$, the S-R pair chooses an action $A_t\in \mathcal{A}$ based on its current state $S_t$.  Next, the probabilistic transition function $\mathcal{P}:\mathcal{S}\times\mathcal{A}\times\mathcal{S}\rightarrow[0,1]$ represents the probability of the next state $S_{t+1}$ given the current state  $S_t$ and action $A_t$. The transition function depends on the jamming strategy and the collision avoidance protocol. We note that the set of possible next states is reduced  depending on the current state and action. If $A_t=h$, then $S_{t+1}\in\{J,I,1\}$ regardless of $S_t$ as shown in Fig. \ref{Fig4}. The transitions for the case of $A_t=s$ are represented in Fig. \ref{Fig5}.
If $A_t=s$ and $S_t\in[1:K_m-1]$, then $S_{t+1}\in\{S_t+1,J\}$ because the S-R pair will try to communicate at time $t+1$ since it has occupied the channel, and $S_{t+1}$ will be $S_{t}+1$ if the communication is successful and $S_{t+1}$ will be $J$ otherwise. If $A_t=s$ and $S_t=K_m$, then $S_{t+1}=J$ from the definition of~$K_m$. Note that $A_t=s$ is not allowed when $S_t\in\{J,I\}$. Now, the S-R pair gets a reward at the end of time $t$, based on $S_t$ and $A_t$, according to the reward function $\mathcal{U}:\mathcal{S}\times\mathcal{A}\rightarrow\mathbb{R}$. If the communication was successful in time $t$, the S-R pair gets a reward $R>0$ proportional to the communication throughput. If the S-R pair tried to communicate but it was jammed, the S-R pair just wasted some communication resource such as transmit power and gets a reward $-L<0$.  If the S-R pair decides to hop at the end of time $t$, it needs to pay a hopping cost $C>0$, because the sender cannot start communication immediately right after it hops to new channel due to the settling time, e.g., the settling time is about 7.6ms in Atheros chipset card \cite{Navda:2007}. Another reason behind the hopping cost is due to the addition of the collision avoidance phase. Note that the collision avoidance phase is skipped if the S-R pair stays in the same channel. 

By taking into account all the aforementioned factors, the reward $U_{t+1}$ at time $t+1$ (right after time $t$) for state $S_t$ and action $A_t$ is given as follows: 
\begin{align}
 U_{t+1}(S_t,A_t)=
 \begin{cases}
 \hfil R \!\!\!\! &\textrm{for \ }S_t\in [1:K_m], A_t=s,\\
 \hfil R-C \!\!\! &\textrm{for \ }S_t\in[1:K_m], A_t=h,\\
 \hfil -L-C \!\!\!\! &\textrm{for \ }S_t= J, A_t=h,\\
 \hfil -C \!\!\!\! &\textrm{for \ }S_t= I, A_t=h.
 \end{cases}\label{eq:1}
\end{align}

\begin{figure}
\centering
\includegraphics[width=0.7\columnwidth]{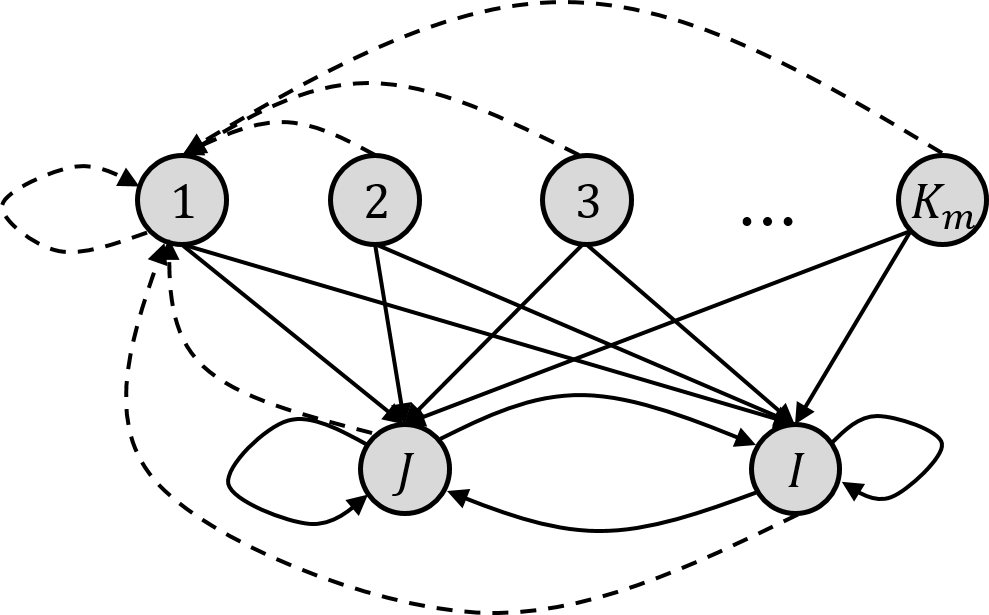}
\caption{Possible state transitions from $S_t$ to $S_{t+1}$ when $A_t=h$. The solid and the dashed lines represent the transitions to $S_{t+1}\in\{J,I\}$ and $S_{t+1}=1$, respectively.}\label{Fig4}
\end{figure}

\begin{figure}
\centering
\includegraphics[width=0.7\columnwidth]{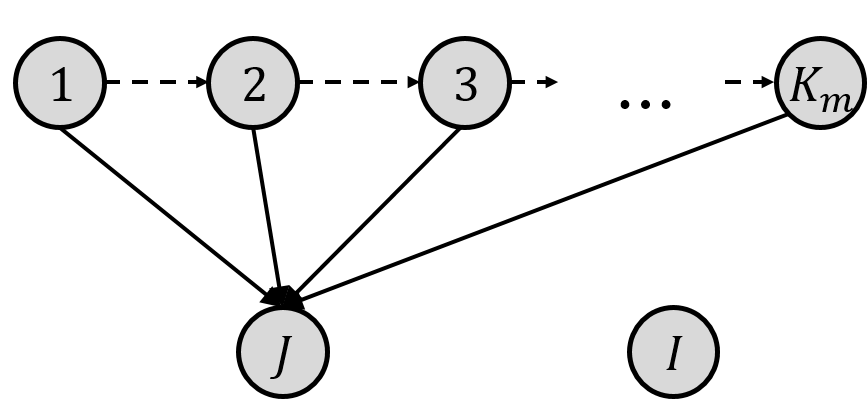}
\caption{Possible state transitions from $S_t$ to $S_{t+1}$ when $A_t=s$. The solid and the dashed lines represent the transitions to $S_{t+1}=J$ and $S_{t+1}\in[1:K_m]$, respectively. }\label{Fig5}
\end{figure}

The policy $\pi:\mathcal{S}\rightarrow\mathcal{A}$ of each S-R pair specifies the action $\pi(S)$ that the S-R pair will choose at state $S$. 
If an S-R pair starts communication at $t=0$, the value function corresponding to policy $\pi$ with initial state $S$ is defined by 
\begin{align}
V_{\pi}(S)=E_{\pi}\left[\sum_{t=0}^\infty\gamma^{t}U_{t+1}(S_t,A_t) \bigg\rvert  S_0=S\right],\label{eq:2}
\end{align}
where $0<\gamma<1$ is the discount factor, the parameter that captures the importance of the current reward compared to the future rewards. Hence, $V_{\pi}(S)$ corresponds to the expected discounted sum of rewards if the S-R pair determines its action based on policy $\pi$ when the initial state $S_0$ is $S$. Similarly, the action-value function corresponding to policy $\pi$, starting with  state $S$ and action $A$, is defined by 
\begin{align}
Q_{\pi}(S,A)=E_{\pi}\left[\sum_{t=0}^\infty\gamma^{t}U_{t+1}(S_t,A_t) \bigg\rvert  S_0=S,A_0=A\right].\label{eq:3}
\end{align}
The value function and the action-value function can be derived using Bellman expectation equation \cite{BELLMAN:1960}.

The maximum value function $V_{\pi}(S)$ and the maximum action-value function $Q_{\pi}(S,A)$ over all possible policies $\pi$ are denoted by $V^*(S)$ and $Q^*(S,A)$, respectively. They have the following relationship:  
\begin{align}
V^*(S)= \max_{A\in{\mathcal{A}}}Q^*(S,A). \label{eq:4} 
\end{align}
If policy $\pi^*$ satisfies $V_{\pi^*}(S)\geq V_{\pi}(S)$ for any policy $\pi$ and state $S\in\mathcal{S}$, then $\pi^*$  is said to be optimal.  An optimal policy $\pi^*$ also satisfies $Q_{\pi^*}(S,A)\geq Q_{\pi}(S,A)$ for any $\pi$, $S\in\mathcal{S}$, and $A\in\mathcal{A}$. It can be easily checked that if a  policy $\pi^*$ satisfies the following  for any $S\in\mathcal{S}$, it is optimal: 
\begin{align}
\pi^*(S)\in \argmax_{A\in{\mathcal{A}}}Q^*(S,A). \label{eq:5} 
\end{align}
Hence, we can find the optimal frequency hopping policy from the maximum action-value function $Q^*(S,A)$.  A standard approach to derive $Q^*(S,A)$ is to use the following Bellman optimality equation \cite{BELLMAN:1960}:
\begin{align}
\begin{split}\label{eq:6} 
&Q^*(S_t,A_t)\\
&\overset{(a)}=U_{t+1}(S_t,A_t)+\gamma\sum_{S_{t+1}\in\mathcal{S}}p(S_{t+1}|S_t,A_t)V^*(S_{t+1})
\end{split}\\
\begin{split}\label{eq:7} 
&\overset{(b)}=U_{t+1}(S_t,A_t)\\
&~~+\gamma\sum_{S_{t+1}\in\mathcal{S}}\max_{A_{t+1}\in\mathcal{A}}p(S_{t+1}|S_t,A_t)Q^*(S_{t+1},A_{t+1}),
\end{split}
\end{align}
where $p(S_{t+1}|S_t,A_t)$ is from the probabilistic transition function $\mathcal{P}$. Here, $(a)$ is since the maximum action value function can be expressed by using the next reward and the probabilistic sum of the next maximum value functions and $(b)$ is due to  \eqref{eq:4}. 
However, Bellman optimality equation does not have a closed form solution due to the non-linearity. Instead, we can derive $Q^*(S,A)$ in an iterative way by using value iteration \cite{BELLMAN:1960}. If the environment (the transition function) of the MDP is not fully known to the S-R pairs, one possible approach to obtain $Q^*(S,A)$  is to use the Q-learning \cite{Watkins:1992}, which approximates the unknown transition function in an empirical way.

Note that the starting state $S_0$ is a random variable for our model, where $S_0\in\{1,J,I\}$. Hence, as a criterion for evaluating a  jamming strategy or a frequency  hopping policy, we use the expected discounted sum of rewards (EDSR) corresponding to the policy $\pi$ defined by  
\begin{align}
\bar{U}_\pi=E_{S_0}\left[E_{\pi}\left[\sum_{t=0}^\infty\gamma^{t}U_{t+1}(S_t,A_t)\bigg\rvert S_0\right]\right].\label{eq:8}
\end{align}
Note that the maximum EDSR 
is achieved by an optimal policy $\pi^*$. 

\section{Arms race in single-user environment}\label{sec4}
In this section, we continue the arms race by characterizing an optimal frequency hopping policy against the basic sweep jamming described in Section \ref{sec2C} for the single-user case. We note that for the single-user case, the reactive sweep jamming induces the same MDP environment with the basic sweep jamming, which will be proved in  the proof of Theorem \ref{thm1}.
In this section, we show that the arms race reaches an equilibrium after the round of finding an optimal frequency hopping policy against the basic or the reactive sweep jamming.

First, the following theorem shows an optimal frequency hopping policy against the basic or the reactive sweep jamming.\footnote{The result of Theorem \ref{thm1} was previously shown for the basic sweep jamming in \cite[Proposition~1]{Wu:2012}.} 
\begin{theorem}\label{thm1}
For the single-user case, an optimal frequency hopping policy $\pi^*$ against the basic or the reactive sweep jammer is a $K^*$-staying policy for some threshold $K^*\in[0:T-1]$, i.e.,
\begin{align}
\pi^*(S)=
 \begin{cases}
 s \mathrm{\ \ \ for \ } S\in[1:K^*],\\
 h \mathrm{\ \ \ otherwise }.\label{eq:9}
 \end{cases}
\end{align}
The threshold $K^*$ is called the optimal staying-threshold. Note that if $K^*=0$, the set $[1:K^*]$ is an empty set and hence the policy is to hop  always.
\end{theorem}

\begin{proof}
Let us first state the transition functions $p(S_{t+1}|S_t,A_t)$ in the presence of the basic or the  reactive sweep jammer. Note that for the single-user case, the set $\mathcal{S}$ of states is given as $\{J, 1,2,\cdots, T-1\}$, i.e., there is no state $I$ because a collision occurs only if there are multiple S-R pairs. For the basic sweep jamming, it is apparent that $K_m=T-1$ since it scans all the $M$ channels in any window of $T$ consecutive time slots. Note that this $K_m$ remains the same for the reactive sweep jamming because its operation is the same with the basic sweep jamming until it finds the S-R pair.

Now, the transition functions when $A_t=h$ are given as follows for both the basic and the reactive sweep jamming strategies:
\begin{align}
p(J|S,h)&={m\over M}, \label{eq:10}\\
p(1|S,h)&={{M-m}\over M},\label{eq:11}
\end{align}
for any $S\in\mathcal{S}$, since $m$ channels out of all $M$ channels are jammed in each time $t$. 
The transition functions  when $A_t=s$ are given as follows for both the jamming strategies:\footnote{Note that when $S_t=J$, $A_t=h$. Hence, we do not consider $S_t=J, A_t=s$.} 
\begin{align}
 p(J|K,s)&={m\over {M-Km}}, \label{eq:12}\\
 p(J|T-1,s)&=1, \label{eq:12.1}\\
 p(K+1|K,s)&=1-{m\over {M-Km}},\label{eq:13}
\end{align}
for $K\in[1:T-2]$, since the jammer already scanned other $mK$ channels during $K$ consecutive successful communications of the S-R pair, and it attacks $m$ channels out of remaining $M-mK$ channels.

To derive an optimal policy based on \eqref{eq:5}, we need to derive the maximum action-value function $Q^*(S,A)$. Since it is in general difficult to obtain a closed-form solution of $Q^*(S,A)$, we prove the theorem from some monotonicity properties of $Q^*(S,A)$. Note that the maximum action-value functions and the optimal policies are the same for both the jamming strategies because their corresponding transition functions are the same.

First we show that $Q^*(K,h)$ is a constant regardless of $K\in [1:T-1]$.  Note that $Q^*(K,h)$ for $K\in[1:T-1]$ is given as 
\begin{align}
Q^*(K,h)&\overset{(a)}=R-C+\gamma(p(J|K,h)V^*(J)+p(1|K,h)V^*(1))\label{eq:14}\\
&\overset{(b)}=R-C+\gamma(p(J|1,h)V^*(J)+p(1|1,h)V^*(1)),\label{eq:15}
\end{align}
where $(a)$ is from the Bellman optimality equation and $(b)$ is since the transition function \eqref{eq:10}-\eqref{eq:11} for $A_t=h$ does not depend on the current state. Hence, $Q^*(K,h)$ has the same value regardless of $K\in[1:T-1]$. 

Similarly, $Q^*(J,h)$ can be written as  
\begin{align}
Q^*(J,h)&=-L-C+\gamma(p(J|1,h)V^*(J)+p(1|1,h)V^*(1)).\label{eq:16}
\end{align}
By subtracting \eqref{eq:15} from \eqref{eq:16}, the following equation holds for $K\in[1:T-1]$:
\begin{align}
Q^*(J,h)=Q^*(K,h)-R-L.\label{eq:17}
\end{align}

Now, let us show that $Q^*(K,s)$ is strictly decreasing in $K\in [1:T-1]$ by induction. To that end, we first show $Q^*(T-2,s)>Q^*(T-1,s)$ as follows:  
\begin{align}
\begin{split}
&Q^*(T-1,s)-Q^*(T-2,s)\\
\overset{(a)}=&\gamma\sum_{S_{t+1}\in\mathcal{S}}(p(S_{t+1}|T-1,s)-p(S_{t+1}|T-2,s))V^*(S_{t+1})
\end{split}\label{eq:18}\\ 
\begin{split}
=&\gamma(V^*(J)-p(T-1|T-2,s)V^*(T-1)\\
&-p(J|T-2,s)V^*(J))
\end{split}\label{eq:19}\\
=&\gamma\cdot p(T-1|T-2,s)(V^*(J)-V^*(T-1))\label{eq:20}\\
\overset{(b)}\leq&\gamma\cdot p(T-1|T-2,s)(V^*(J)-Q^*(T-1,h))\label{eq:21}\\
\overset{(c)}=&\gamma\cdot p(T-1|T-2,s)(-R-L)\label{eq:22}\\
<&0, \label{eq:23}
\end{align}
where $(a)$ is from the Bellman optimality equation, $(b)$ is due to \eqref{eq:4}, and $(c)$ is by \eqref{eq:17} and $Q^*(J,h)=V^*(J)$. 

Next, under the assumption $Q^*(K-1,s)>Q^*(K,s)$ for $K\in[3:T-1]$, $Q^*(K-2,s)>Q^*(K-1,s)$ can be proved as follows:   
\begin{align}
\begin{split}\label{eq:24}
Q&^*(K-1,s)-Q^*(K-2,s)\\
\overset{(a)}=&\gamma\sum_{S_{t+1}\in\mathcal{S}}(p(S_{t+1}|K-1,s)-p(S_{t+1}|K-2,s))V^*(S_{t+1})
\end{split}\\  
\begin{split}\label{eq:25}
=&\gamma(p(K|K-1,s)V^*(K)+p(J|K-1,s)V^*(J)\\
&-p(K-1|K-2,s)V^*(K-1)-p(J|K-2,s)V^*(J))
\end{split}\\
\begin{split}\label{eq:26}
=&\gamma\cdot(((p(K-1|K-2,s)-p(K|K-1,s))(V^*(J)\\
&-V^*(K-1))+p(K|K-1,s)(V^*(K)-V^*(K-1)))
\end{split}\\
\overset{(b)}<&\gamma\cdot p(K|K-1,s)(V^*(K)-V^*(K-1))\label{eq:27}\\
\overset{(c)}\leq&0, \label{eq:28}
 \end{align}
where $(a)$ is by Bellman optimality equation, $(b)$ is by \eqref{eq:4} and \eqref{eq:17}, and $(c)$ is because
$Q^*(K-1,s)>Q^*(K,s)$ from the assumption and $Q^*(K-1,h)=Q^*(K,h)$. By induction, we conclude that $Q^*(K,s)$ is strictly decreasing in $K$.

Since $Q^*(K,h)$ and $Q^*(K,s)$ are constant and strictly decreasing in $K$, respectively, there exists a threshold $K^*\in[1:T-1]$ such that $Q^*(K,s)\geq Q^*(K,h)$ for $K\in[1:K^*]$ and  $Q^*(K,s)\leq Q^*(K,h)$ for  $K\in[K^*+1:T-1]$, unless $Q^*(1,s)< Q^*(1,h)$. If  $Q^*(1,s)< Q^*(1,h)$, let $K^*=0$. From \eqref{eq:5}, we conclude that an optimal frequency hopping policy is staying for $K\in[1:K^*]$ and hopping for $K\in[K^*+1:T-1]$, with the interpretation that if $K^*=0$, the set $[1:K^*]$ is an empty set and hence the policy is to hop for all $K\in [1:T-1]$.
\end{proof}
The optimality of the $K^*$-staying policy can be explained as follows. The probability that an S-R pair is jammed increases as it stays the same channel longer. Hence, it is better to hop to the next channel if the risk of being jammed by staying the same channel outweighs the cost due to the hopping action. The optimal staying-threshold $K^*$ is a critical state at which staying is more beneficial to the S-R pair than hopping, and it can be obtained by deriving the maximum action-value function through value iteration and then applying \eqref{eq:5}. 

According to Theorem \ref{thm1}, as a next round of the arms race against the basic or the reactive sweep jamming, the S-R pair employs the $K^*$-staying policy. Now we show that the arms race has reached an equilibrium at this stage, by proving that the next round of the arms race does not change the EDSR. To that end, we first show that an optimal jamming strategy against the $K^*$-staying policy is the $K^*$-memory jamming, which is called ``smarter jamming" in \cite{Wu:2012}.

\begin{theorem}\label{thm1.1}
For the single-user case, an optimal jamming strategy against the $K^*$-staying policy with optimal staying-threshold $K^*$ is the $K^*$-memory jamming strategy.
\end{theorem}
\begin{proof}
For a given frequency hopping policy, note that the EDSR  decreases as the probability that the S-R pair is jammed increases for each pair of state and action. Hence, an optimal jamming strategy against a frequency hopping policy will scan the most probable set of $m$ channels where the S-R pair would communicate. The jammer will judge such a set of $m$  channels based on the frequency hopping policy of the S-R pair and the history of the previously scanned channels. For the $K^*$-staying policy, note that only the history of the recent $K^*$ time slots is needed since the S-R pair stays in a channel at most $K^*$ consecutive time slots.

Without loss of generality, consider the situation where the jammer chooses the most probable set of $m$ channels at time $K^*+1$, under the assumption that the jammer has scanned channels $\{c_1,c_2,...,c_{K'}\}$ during time slots  $[1:K^*]$, where $K'\leq mK^*$.  We define the following events: 
\begin{itemize}
\item $E(c_j)$: Event that the S-R pair occupies $c_j$ at time slot $K^*+1$.
\item $E_{\tau}(c_j)$ for $\tau\in[0:K^*]$: Event that the S-R pair newly hopped to channel $c_j$  at the end of time $K^*-\tau$ and then has consecutively chosen the staying action and occupied channel $c_j$ from time slot $K^*+1-\tau$ to time slot $K^*+1$. 
\end{itemize}
Now, we claim that $P(E(c_j))\leq P(E(c_i))$ for $j\in [1:K']$ and $i\in [K'+1:M]$. Note that  
\begin{align}
P(E(c_\ell))= \sum_{\tau=0}^{K^*} P(E_{\tau}(c_\ell))\label{eq:28.1}  
\end{align}
for any $\ell\in [1:M]$. For $\tau\in [1:K^*]$, $ P(E_{\tau}(c_j))\leq  P(E_{\tau}(c_i))$  for $j\in [1:K']$ and $i\in [K'+1:M]$, because the jammer scanned channels $[1:K']$ for time slot $[1:K^*]$ and hence it is more difficult for channel $\{c_1,...,c_{K'}\}$ to have consecutive successful communications. For $\tau=0$, $P(E_{\tau}(c_\ell))$ is trivially constant for $\ell\in [1:M]$ since when an S-R pair determines to hop, the next channel is selected uniformly at random among all $M$ channels. Thus, we conclude that $P(E(c_j))\leq P(E(c_i))$ for any $j\in [1:K']$ and $i\in [K'+1:M]$.

Also, $P(E(c_i))$ is a constant for all $i\in [K'+1:M]$ due to the symmetry of the $K^*$-staying policy. Hence, an optimal set of $m$ channels to scan at time  $K^*+1$ is obtained by arbitrarily selecting $m$ channels among channels $\{c_{K'+1},...,c_M\}$. Note that we can always choose non-overlapping $m$ channels in $\{c_{K'+1},...,c_M\}$ since $K'\leq m(T-1)$ and $M=mT$. Hence, we conclude that an optimal jamming strategy against the $K^*$-staying policy is the $K^*$-memory jamming strategy.
\end{proof}
This theorem implies that as a next step of the arms race against the $K^*$-staying policy, the jammer can try the $K^*$-memory jamming. It was numerically shown in \cite{Wu:2012} that the $K^*$-memory jamming decreases the EDSR, compared to the basic sweep jamming, for the $K^*$-staying policy. Here we analytically show that for the $K^*$-staying policy, the EDSR  is the same for the basic or the reactive sweep jamming  and for the $K^*$-memory jamming attacks.

\begin{theorem}\label{thm2}
In the single-user environment, for the $K^*$-staying policy with optimal staying-threshold $K^*$, the EDSRs are all the same for the basic sweep jamming, the reactive sweep jamming, and the $K^*$-memory jamming attacks.
\end{theorem}
\begin{proof}
The EDSR for a given pair of frequency hopping policy and jamming strategy can be derived from the Bellman expectation equation \cite{BELLMAN:1960}, which only requires the transition functions associated with \emph{visitable} state-action pairs. For the $K^*$-staying policy, the Bellman expectation equations corresponding to the three jamming strategies require only the transition functions $p(\cdot|J,h)$, $p(\cdot|K,s)$ for $K\in[1:K^*]$, and $p(\cdot|K^*+1,h)$. Note that these transition functions are the same regardless of the jamming strategies. It was shown in Theorem \ref{thm1}  that the transition functions of the basic or the reactive sweep jamming are the same. Furthermore, the transition functions of the basic or the reactive sweep jamming and the $K^*$-memory jamming are the same since they attack non-overlapping channels in any window of  $K^*$ consecutive time slots. Hence, we conclude that the EDSR is the same in the three jamming attacks.
\end{proof}
From Theorems  \ref{thm1}-\ref{thm2}, we can finally obtain the following main theorem on the equilibrium of the arms race for the single-user case:
\begin{theorem} \label{thm:main_single}
For the single-user case, the $K^*$-staying policy with optimal staying-threshold $K^*$ and the basic or the reactive sweep jamming strategy establish an equilibrium. 
\end{theorem} 
\begin{proof} 
The proof is immediate from Theorems  \ref{thm1}-\ref{thm2} since 
\begin{enumerate}
\item Theorem \ref{thm1} implies that no frequency hopping policy is better than the $K^*$-staying policy, against the basic or the reactive sweep jamming and 
\item Theorems \ref{thm1.1} and \ref{thm2} mean that no jamming strategy is stronger than the basic or the reactive sweep jamming, against the $K^*$-staying policy. 
\end{enumerate}
\end{proof} 

Fig. \ref{Fig6} illustrates the arms race between the S-R pair and the jammer, where we can check that the $K^*$-memory jammer does not decrease the EDSR, compared to the basic or reactive sweep jammer, against the $K^*$-staying policy. 
The optimal staying-threshold $K^*$ depends on various model parameters. It increases in $M$ and decreases in $m$ since $p(J|K,s)$ decreases in $M$ and increases in $m$ for any $K\in[1:T-1]$. Also, it decreases in both $R$ and $L$ since large $R$ and $L$ mean that the difference in the rewards for the jammed and successful communication is large. Finally, $K^*$ increases in $C$ since large hopping cost means that the received reward when the S-R pair hops to the next channel is small. In Fig. \ref{Fig7}, the tendency of $K^*$ in $L$ and $m$ in the basic or reactive sweep jamming environment is shown.

\begin{figure}
\centering
\includegraphics[width=0.8\columnwidth]{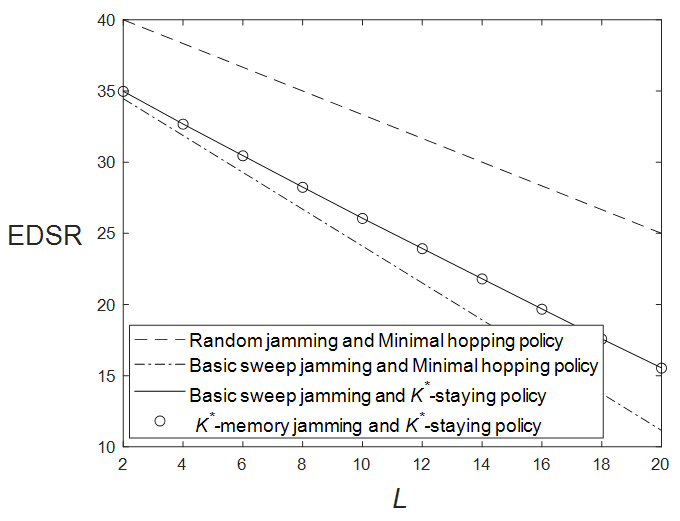}
\caption{The EDSRs according to jamming cost~$L$ in the arms race between the user and the jammer for parameters $R=5$, $C=5$, $n=1$, $M=60$, $m=5$, and $\gamma=0.9$.}\label{Fig6}
\end{figure}

\begin{figure}
\centering
\includegraphics[width=0.8\columnwidth]{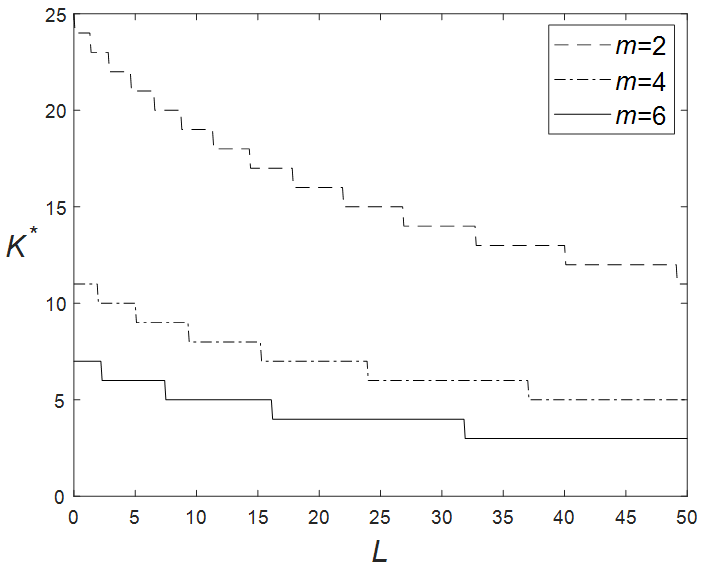}
\caption{Optimal staying-threshold $K^*$ versus jamming cost~$L$ for $m=2,4,6$ for the basic or reactive sweep jamming environment with parameters by $R=5$, $C=5$, $n=1$, $M=60$, and $\gamma=0.9$.}\label{Fig7}
\end{figure}

\begin{remark}\label{rmk1}
Note that the jammer can choose either the basic or the reactive sweep jamming attacks at the equilibrium, i.e., both result in the same EDSR. However, for the basic sweep jamming, the S-R pair can eventually learn the sweep pattern after a sufficiently long time. If the S-R pair gets to know the sweep pattern, it can significantly increase the EDSR by staying a channel as long as possible and hopping to other channels right before being jammed. In the reactive sweep jamming, it is difficult for the S-R pair to learn the sweep pattern since the sweep pattern is initialized frequently. Hence, for the single-user case, it is preferable for the jammer to choose the reactive sweep jamming.
\end{remark}

\section{Arms race in multi-user environment}\label{sec5}
In this section, we consider the multi-user case $(n \geq 2)$ where two or more S-R pairs communicate in the presence of a jammer. We continue the arms race by analyzing an optimal frequency hopping policy against the basic sweep jamming described in Section \ref{sec2C}.\footnote{For the single-user case, the reactive sweep jamming and the basic sweep jamming induce the same MDP environment as shown in Section \ref{sec4}. In contrast, for the multi-user case with $n\geq 2$, the reactive sweep jamming  induces a different  MDP environment compared to the basic sweep jamming, as explained in Section~\ref{sec6}.} 

In multi-user environment, each S-R pair has an inactive probability $0\leq\theta\leq 1$, defined as the probability that an S-R pair who just hops to a channel is not allowed to communicate in the communication phase. The inactive probability of an S-R pair depends on various factors including frequency hopping policy, collision avoidance protocol, and jamming strategy. Note that  an inactive probability affects the MDP environment from which an optimal frequency hopping policy can be derived, and in turn, the updated frequency hopping policy affects the inactive probability. Such a recurrent relationship makes the problem of finding  an optimal policy $\pi^*$ against the basic sweep jamming non-trivial, and  necessitates an iterative approach. Hence, to find an optimal policy $\pi^*$ and the resultant inactive probability $\theta^*$ against the basic sweep jamming for a given collision avoidance protocol, we start with an initial pair $(\pi_1, \theta_1)$, where $\pi_1$ is a $K$-staying policy with an arbitrarily chosen staying-threshold and $\theta_1$ is the corresponding inactive probability. Next, we can find a frequency hopping policy $\pi_2$ from Bellman optimality equation based on the transition function assuming $\theta_1$, and let $\theta_2$ denote the inactive probability corresponding to $\pi_2$. This iterative update proceeds and let $f$ and $g$ denote the update functions, i.e., $\pi_{i+1}=f(\theta_i)$ and $\theta_{i+1}=g(\pi_{i+1})$. Note that $\pi^*$ is an optimal policy if $\pi^*=f\circ g (\pi^*)$. 

The following theorem characterizes $\pi=f(\theta)$ against the basic sweep jamming, for any inactive probability $\theta$. 
\begin{theorem}\label{thm3}
For the basic sweep jamming, an optimal frequency hopping policy $\pi=f(\theta)$ under the environment described by an inactive probability $\theta$ is a $K(\theta)$-staying policy for some threshold $K(\theta)\in[0:T-1]$, i.e., 
\begin{align}
\pi(S)=
 \begin{cases}
 s \mathrm{\ \ \ for \ } S\in[1:K(\theta)]\\
 h \mathrm{\ \ \ otherwise }
 \end{cases}.\label{eq:29}
\end{align}
Note that if $K(\theta)=0$, the set $[1:K(\theta)]$ is an empty set and hence the policy is to hop always. 
\end{theorem}

\begin{proof}
Similar to the proof of Theorem \ref{thm1}, we start by specifying the transition function with respect to each S-R pair. In multi-user environment, the set $\mathcal{S}$ of states for each S-R pair includes state $I$ since a collision occurs if more than or equal to two S-R pairs exist in the same channel. The transition function of each S-R pair for $A_t=h$ is given as follows:
\begin{align}
p_\theta(I|S,h)&=\theta, \label{eq:30}\\
p_\theta(J|S,h)&=(1-\theta){m\over M}, \label{eq:31}\\
p_\theta(1|S,h)&=(1-\theta){{M-m}\over M},\label{eq:32}
\end{align}
for any $S\in\mathcal{S}$. Here, \eqref{eq:30}-\eqref{eq:32} is proved similarly as \eqref{eq:10}-\eqref{eq:11}, except that an S-R pair who just hops to a channel get the chance to communicate in the communication phase with probability $1-\theta$. The transition function for $S_t=K\in [1:T-1]$ and $A_t=s$ is the same with \eqref{eq:12}-\eqref{eq:13} because any state does not transit to state $I$ if it could take the staying action, i.e., an S-R pair can take the staying action only if its communication was successful and hence  it has the priority to communicate in the communication phase if it stays at the same channel.  

A policy obtained from Bellman optimality equation and \eqref{eq:5} using the described transition function can be derived similarly as Theorem \ref{thm1}, but now we need to consider state $I$. We prove the theorem from some monotonicity properties of the maximum action-value function for the given $\theta$, $Q_\theta^*(S,A)$. We first show that $Q_\theta^*(K,h)$ is a constant regardless of $K\in [1:T-1]$.  Note that $Q_\theta^*(K,h)$ for $K\in[1:T-1]$ and $Q_\theta^*(J,h)$ are given as 
\begin{align}
Q_\theta^*(K,h)&=R-C+\gamma\!\!\!\!\!\!\sum_{S_{t+1}\in\{1,J,I\}}\!\!\!\!\!\!p_\theta(S_{t+1}|1,h)V_\theta^*(S_{t+1}),\label{eq:33}\\
Q_\theta^*(J,h)&=-L-C+\gamma\!\!\!\!\!\!\sum_{S_{t+1}\in\{1,J,I\}}\!\!\!\!\!\!p_\theta(S_{t+1}|1,h)V_\theta^*(S_{t+1}),\label{eq:34}
\end{align}
respectively, which are proved similarly as in \eqref{eq:14}-\eqref{eq:15}. By \eqref{eq:33}, $Q_\theta^*(K,h)$ is a constant regardless of $K\in [1:T-1]$. By subtracting \eqref{eq:33} from \eqref{eq:34}, \eqref{eq:17} also holds in multi-user environment for $K\in[1:T-1]$.

The strictly decreasing property of $Q_\theta^*(K,s)$ in $K\in[1:T-1]$ can be proved in the exactly same way as in the proof of Theorem \ref{thm1}, because any state does not transit to state $I$ in the staying action and \eqref{eq:17}, used to prove strictly decreasing property of $Q_\theta^*(K,s)$, also holds in multi-user environment. Since $Q_\theta^*(K,h)$ and $Q_\theta^*(K,s)$ are constant and strictly decreasing in $K$, respectively, a policy obtained by Bellman optimality equation in inactive probability $\theta$ is the $K(\theta)$-staying policy with threshold $K(\theta)\in[0:T-1]$.
\end{proof}

This theorem implies that for any inactive probability $\theta$, an optimal frequency hopping policy against the basic sweep jamming is a $K(\theta)$-staying policy. Now, the following theorem shows that an optimal policy $\pi^*$ and the resultant  inactive probability $\theta^*$, i.e., $\pi^*=f(\theta^*)$ and $\theta^*=g (\pi^*)$, can be obtained through a finite number of iterations for any initial $K$-staying policy $\pi_1$.

\begin{theorem}\label{thm4}
For any initial $K$-staying policy $\pi_1$ with arbitrarily chosen $K\in[0:T-1]$, there exists $s\in [0:T-1]$ such that 
\begin{align}
(f\circ g)^s(\pi_1)=(f\circ g)^{s+1}(\pi_1), \label{eq:35}
\end{align}
which implies that an optimal frequency hopping policy against the basic sweep jamming is a $K^*$-staying policy, and the optimal staying-threshold $K^*$ can be obtained by at most $T-1$ iterative updates. 
\end{theorem}          

\begin{proof}
Let us first show that $g\circ f$ is an increasing step function under the assumption that  $g\circ f$ is an increasing function. This assumption will be justified at the end of this proof. Since a policy induced by an update function $f$ is the $K$-staying policy with staying-threshold $K\in[0:T-1]$ as shown in Theorem~\ref{thm3}, the cardinality of set $\{g\circ f(\theta)|0 \leq\theta\leq 1\}$ is equal to or less than $T$. Hence, $g\circ f$ is an increasing step function where the number of the steps is upper-bounded by $T$. We assume that the step function has the following $i\leq T$ steps: $[1,\theta'_{1}),[\theta'_{1},\theta'_{2}),...,[\theta'_{i},1]$. Step $[\theta'_{k-1},\theta'_{k})$ $(k\leq i)$ corresponds to policy $\pi'_k$ and the different step induces the different policy.

Next, we derive that an optimal frequency hopping policy is obtained through less than $T$ updates of a pair $(\pi,\theta)$ by using $f(\theta)$ and $g(\pi)$. Let the update start in $(\pi_1,\theta_1)$ and $(\pi_j,\theta_j)$ is updated to $(\pi_{j+1},\theta_{j+1})$, i.e., $\pi_{i+1}=f(\theta_i)$ and $\theta_{i+1}=g(\pi_{i+1})$. For convenience, we denote the update as $(\pi_j,\theta_j)\rightarrow(\pi_{j+1},\theta_{j+1})$. If  $\theta_j\geq\theta_{j+1}$ (or $\leq$), then $\theta_{j+1}\geq\theta_{j+2}$ (or $\leq$) since $g\circ f(\theta_j)\geq g\circ f(\theta_{j+1})$ (or $\leq$). Hence, $\theta_j$ is increasing or decreasing on $j$. By monotonic property of $\theta_j$, $k(j)$ corresponding to $\pi_j=\pi'_{k(j)}$ also monotonic on $j$. Since an optimal policy holds $(f\circ g)(\pi^*)=\pi^*$ and $i\leq T$, there exist $s\in [0:T-1]$ such that $(f\circ g)^s(\pi_1)=(f\circ g)^{s+1}(\pi_1)$ for any initial $K$-staying policy $\pi_1$. 

Now, it remains to show the increasing property of $g\circ f$. It can be checked as follows.  First, Theorem \ref{thm3}  shows that $f(\theta)$ is a $K(\theta)$-staying policy. This threshold $K(\theta)$ increases in $\theta$ since the risk of not getting the chance to communicate in the communication phase right after the hopping action increases in $\theta$. It is also shown in Fig. \ref{Fig8}. Next, for $K$-staying policy, $g(\pi)$ increases in $K$. 
To see this, first note that the number of the S-R pairs selecting hopping action decreases in $K$ and hence the probability that multiple S-R pairs are in the same channel also decreases. Hence, as $K$ increases, the total number of channels occupied by the S-R pairs increases. Therefore, if an S-R pair takes the hopping action, the probability that it hops to a channel that has been already occupied by other S-R pair(s), i.e., the probability of collision $g(\pi)$, also increases in $K$. Consequently, the function $g\circ f$ is  increasing.\end{proof}

\begin{remark}\label{rmk2}
In the single-user case $(n=1)$, we can obtain an optimal frequency hopping policy $\pi^*$ through a single update ($s=1$ in Theorem \ref{thm4}) because $\theta=0$.
\end{remark}

\begin{figure}
\centering
\includegraphics[width=0.8\columnwidth]{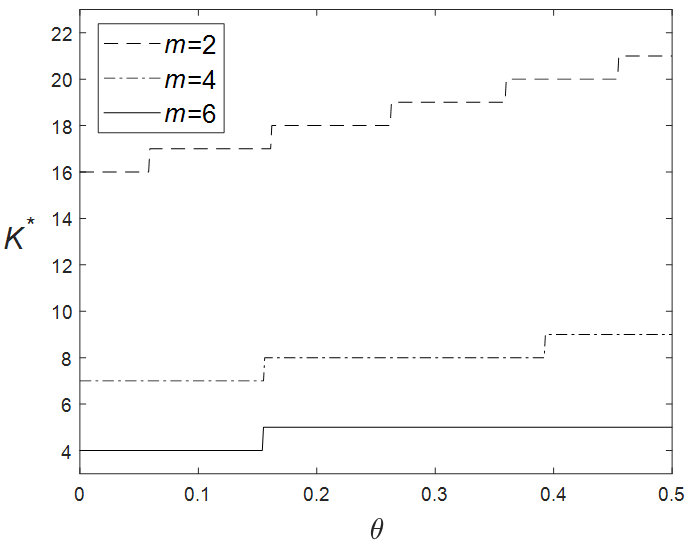}
\caption{Optimal staying-threshold $K(\theta)$ versus inactive probability $\theta$ for $m=2,4,6$ for the basic sweep jamming environment with parameters $R=5$, $C=5$, $L=20$, $M=60$, and $\gamma=0.9$.}\label{Fig8}
\end{figure}

The optimality of the $K^*$-staying policy can be explained similarly as in the single-user case: an S-R pair is better to hop to the next channel if the risk of being jammed by staying the same channel outweighs the cost of hopping together with the risk of not getting the chance to communicate in the communication phase in the newly hopped channel. 

Now the following theorem states that not only the $K^*$-staying policy is an optimal frequency hopping policy against the basic jamming strategy, but also the basic jamming strategy is an optimal jamming strategy against the 
$K^*$-staying policy. 
\begin{theorem}\label{thm5}
In the multi-user case,  against the $K^*$-staying policy with optimal staying-threshold $K^*$,  the basic sweep jamming is an optimal jamming strategy. 
\end{theorem}
\begin{proof}
It can be proved in a similar manner as in~Theorem~\ref{thm1.1} that an optimal jamming strategy against the $K^*$-staying policy is  the $K^*$-memory jamming also in the general multi-user case. Now let us show that the EDSR of  an S-R pair  is the same for both the basic sweep jamming and the $K^*$-memory jamming environments, which proves the theorem. The EDSR of an S-R pair for a given pair of frequency hopping policy and jamming strategy can be derived similarly with Theorem \ref{thm2}, but we also consider whether the inactive probability is the same in both the jamming strategies. For the $K^*$-staying policy, the Bellman expectation equations corresponding to both the jamming strategies require only the transition functions $p_\theta(\cdot|I,h)$, $p_\theta(\cdot|J,h)$, $p_\theta(\cdot|K,s)$ for $K\in[1:K^*]$, and $p_\theta(\cdot|K^*+1,h)$. These transition functions are the same in both the jamming strategies since they attack non-overlapping channels in any window of $K^*$ consecutive time slots and hence the inactive probability is the same in both the jamming strategies. Consequently, we conclude that the EDSR of an S-R pair is the same in both the jamming environments.
\end{proof}

From Theorems \ref{thm3}-\ref{thm5}, we obtain the following main theorem on the equilibrium of the arms race in general multi-user environment:
\begin{theorem} \label{thm:main_multi}
For general multi-user environment, the $K^*$-staying policy with optimal staying-threshold $K^*$ and the basic sweep jamming strategy establish an equilibrium. 
\end{theorem} 
\begin{proof}
It can be proved in a similar way as in Theorem \ref{thm:main_single}, based on Theorems \ref{thm3}-\ref{thm5}. 
\end{proof}
\begin{remark}\label{rmk2.1}
Consider a cognitive network setting with multiple primary users and secondary users, where each of the secondary users  is allowed to communicate only when there is no primary user in the current channel \cite{Liang:2011}. For the anti-jamming game between the secondary users and a jammer, it can be checked that Theorem \ref{thm:main_multi} also holds. 
\end{remark}

Now, let us provide some examples of collision avoidance protocols. To describe the protocols, assume that the sensing phase was silent for a channel, i.e., there was no S-R pair that has occupied the channel. The simplest protocol would be all-hopping protocol.\footnote{This protocol can be implemented by modifying Carrier Sense Multiple Access (CSMA) with collision detection protocol  \cite{Tanenbaum:2010}.} In this protocol, every sender who newly hops to the channel broadcasts a pilot signal at a random moment. If there are more than one pilot signals, all the S-R pairs give up the communication and randomly hop to the other channels in the next time slot. Although this protocol is simple, it has the disadvantage that no one uses the channel if there are more than two S-R pairs hopping to the channel. The random protocol can resolve this issue.\footnote{This protocol can be implemented by modifying CSMA with collision detection and non-persistent CSMA protocols in \cite{Tanenbaum:2010}.} In the random protocol, each sender broadcasts a pilot signal at a random instant during the collision avoidance phase, and the S-R pair who first broadcasts the signal occupies the channel, i.e., it is allowed to communicate.

In the multi-user case, the optimal staying-threshold $K^*$ at the equilibrium depends on the number $n$ of S-R pairs and the collision avoidance protocol, in addition to various parameters such as $m$, $M$, $R$, $L$, and $C$. It increases in $n$ since the inactive probability increases in $n$. To explain the relationship between $K^*$ and the collision avoidance protocols of interest, we present some bounds on the inactive probability for each collision avoidance protocol in the following propositions.

\begin{proposition}\label{pro1}
For the all-hopping collision avoidance protocol, the inactive probability $\theta$ of a user is bounded as follows for any frequency hopping policy and for any jamming strategy:
\begin{align*}
1-\left(1-{1/M}\right)^{n-1}\leq\theta\leq{{n-1}\over M}. 
\end{align*}
\end{proposition}    

\begin{proof}
To derive a range of $\theta$, we need to state the set of the actions of the other S-R pairs inducing the largest and the smallest inactive probabilities. When an S-R pair hops to the next channel, note that the smaller the number of channels where the other S-R pairs exist is, the smaller the inactive probability is. Hence, a lower-bound on $\theta$ can be derived under the scenario where all the other S-R pairs select the hopping action because when the number of hopping S-R pairs increases, the probability that two or more S-R pairs are assigned to the same channel also increases. Hence, $\theta$ is lower-bounded as: 
\begin{align}
\theta\geq 1-\left(1-{1/M}\right)^{n-1}, \label{eq:36}
\end{align}
since the probability that the S-R pair and another S-R pair are allocated to the different channels is $1-1/M$ and the S-R pairs hop with i.i.d. process. 

On the other hand, an upper-bound on $\theta$ can be derived under the scenario where the other S-R pairs select the staying action, i.e, at most one S-R pair occupies a channel. Thus, $\theta$ is upper-bounded as:
\begin{align}
\theta\leq {{n-1}\over M}, \label{eq:37}
\end{align}
since $(n-1)$ channel are already allocated to the other S-R pairs.
\end{proof}

\begin{proposition}\label{pro2}
For the random collision avoidance protocol,  the inactive probability $\theta$ of a user is bounded as follows for any frequency hopping policy and for any jamming strategy:
\begin{align*}
\sum_{i=1}^{n-1} {i\over {i+1}}{{n-1} \choose i}\left({1/M}\right)^{i}\left(1-{1/M}\right)^{n-1-i}\leq\theta\leq{{n-1}\over M}.
\end{align*}
\end{proposition}   

\begin{proof}
The proof of the lower-bound is similar as in Proposition \ref{pro1}, but now the random protocol is used to avoid the communication collision:  
\begin{align}
\theta&\overset{(a)}\geq\sum_{i=1}^{n-1} {i\over {i+1}}\cdot p_i \label{eq:38}\\ 
&=\sum_{i=1}^{n-1} {i\over {i+1}}{{n-1} \choose i}\left({1/M}\right)^{i}\left(1-{1/M}\right)^{n-1-i},\label{eq:39}
\end{align}
where $p_i$ is the probability that the S-R pair and $i$ other S-R pairs are allocated to the same channel. Here, $(a)$ is because when $i+1$ S-R pairs are allocated to the same channel, the probability that the S-R pair is allowed to communicate in the communication phase is $i\over{i+1}$. The proof of an upper-bound on $\theta$ is same with Proposition \ref{pro1} since we give priority to communicate to the S-R pairs that already occupied their channels, i.e., the S-R pairs choosing the staying action.
\end{proof}

We note that the random protocol has a smaller lower bound than the all-hopping protocol. This is because for the random protocol, when there are multiple S-R pairs hopping to an empty channel, an S-R pair still uses the channel.  When the inactive probability $\theta$ is bounded as $\theta_\mathrm{min}\leq \theta\leq \theta_\mathrm{max}$, the optimal staying-threshold $K^*$ is bounded as:
\begin{align}
K(\theta_\mathrm{min})\leq K^*\leq K(\theta_\mathrm{max}),\label{eq:39.1}
\end{align}
where $K(\theta)$ is a staying-threshold obtained by Bellman optimality equation for inactive probability $\theta$. This inequality holds since $K(\theta)$ increases in $\theta$ as explained in the proof of Theorem \ref{thm4}. Fig. \ref{Fig9} illustrates these bounds on $K^*$ for the all-hopping and the random protocols, from which we can see that the lower and upper bounds are quite tight. 

\begin{figure}
\centering
\includegraphics[width=0.8\columnwidth]{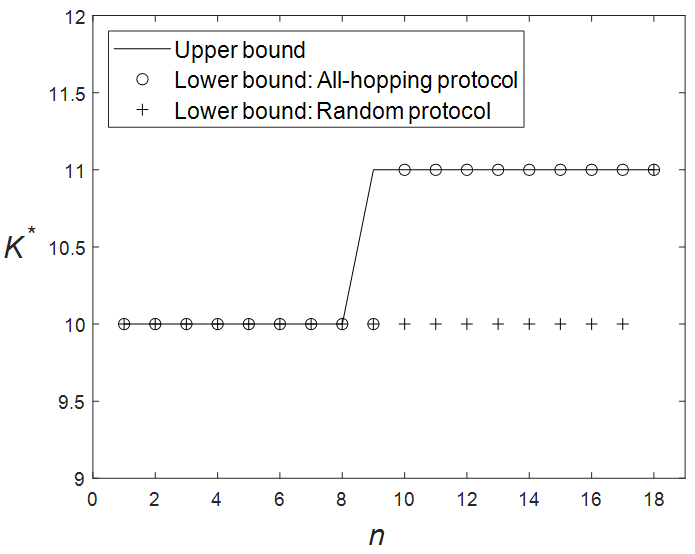}
\caption{The bounds on optimal staying-threshold $K^*$ according to the number $n$ of users for the all-hopping and the random collision avoidance protocols with parameters $R=5$, $C=5$, $L=20$, $M=60$, $m=3$, and $\gamma=0.9$.}\label{Fig9}
\end{figure}

Fig. \ref{Fig10} plots $g\circ f:[0,1]\rightarrow [0,1]$ for the basic sweep jamming environment with the random protocol, in which we can check that the function is an increasing step function as analyzed in the proof of Theorem \ref{thm4}. Fig. \ref{Fig11} illustrates the arms race between an S-R pair and the jammer for the all-hopping protocol, where we can check that the $K^*$-staying policy and the basic sweep jamming strategy establish an equilibrium for  general multi-user case since the $K^*$-memory jammer does not decrease the EDSR compared to the basic sweep jammer, as proved in Theorem \ref{thm5}.
Finally, Fig. \ref{Fig12} shows the EDSRs at the equilibrium for the all-hopping protocol and the random protocol. We can see that EDSR using the random protocol is higher, and the gap from the EDSR using the all-hopping protocol increases as the number of S-R pairs increases. This is because, as the number of S-R pairs increases, the probability that there are multiple S-R pairs in the same channel increases and hence the use of an efficient collision avoidance protocol becomes more critical.

\begin{figure}
\centering
\includegraphics[width=0.8\columnwidth]{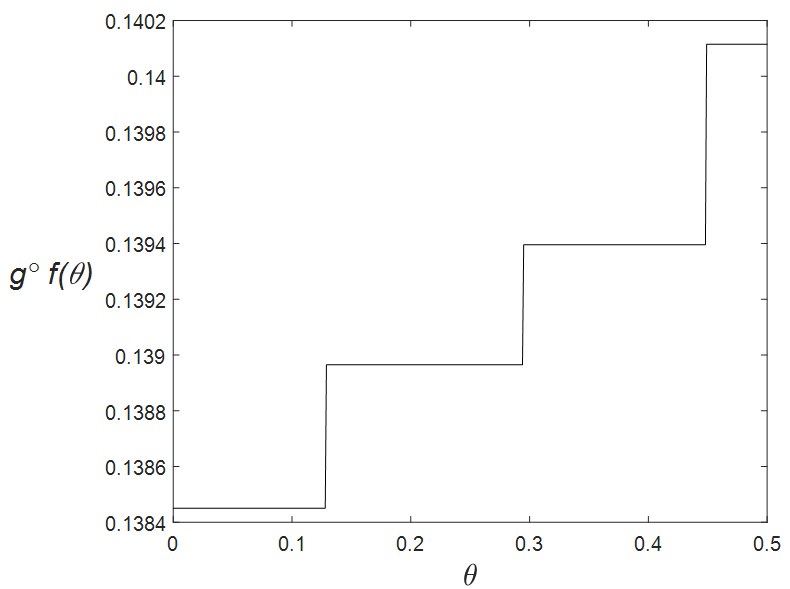}
\caption{Inactive probability $\theta$ versus $g\circ f(\theta)$ for the basic sweep jamming environment with the random protocol and parameters  $R=5$, $C=5$, $L=20$, $n=10$, $M=60$, $m=3$, and $\gamma=0.9$.}\label{Fig10}
\end{figure}

\begin{figure}
\centering
\includegraphics[width=0.8\columnwidth]{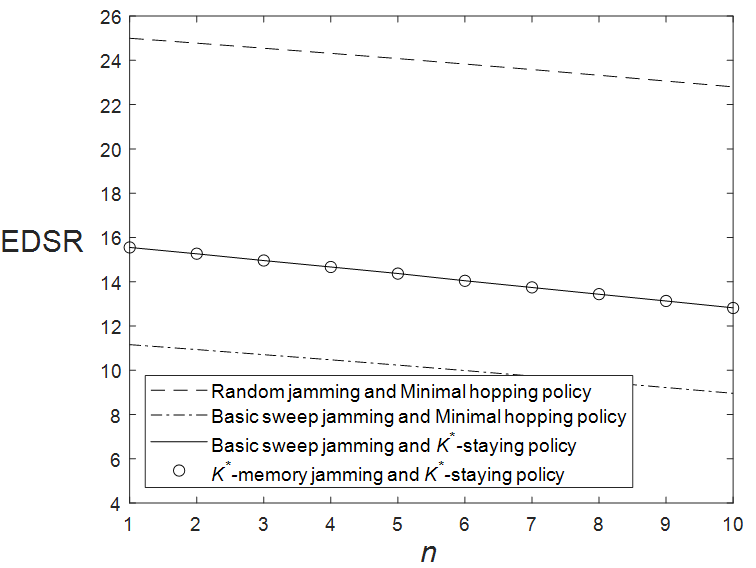}
\caption{The EDSRs  in the arms race according to the number $n$ of users under the all-hopping collision avoidance protocol and parameters $R=5$, $C=5$, $L=20$, $M=60$, $m=5$, and $\gamma=0.9$. }\label{Fig11}
\end{figure}

\begin{figure}
\centering
\includegraphics[width=0.8\columnwidth]{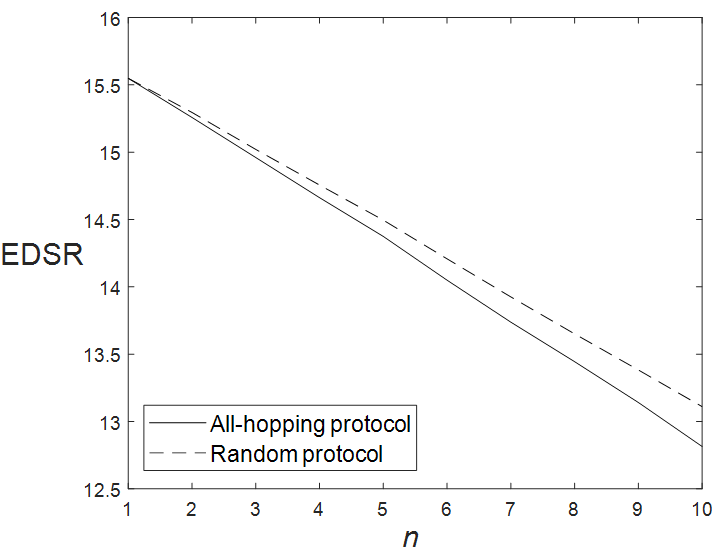}
\caption{The EDSRs according to the number $n$ of users at the equilibrium for the all-hopping and the random collision avoidance protocols with the parameters $R=5$, $C=5$, $L=20$, $M=60$, $m=5$, and $\gamma=0.9$.}\label{Fig12}
\end{figure}

\begin{remark}\label{rmk3}
For the basic sweep jamming, the S-R pairs can eventually learn the sweep pattern after a sufficiently long time as pointed out in Remark \ref{rmk1}. In Section \ref{sec6}, we consider  about the case where the jammer adopts some unpredictable jamming strategies. \end{remark}

\section{Unpredictable jamming strategies in multi-user environment}\label{sec6}
In Section \ref{sec5}, it is shown that the basic sweep jamming and the $K^*$-staying policy with optimal staying-threshold $K^*$ form an equilibrium. One drawback of the basic sweep jamming is that the S-R pairs may eventually learn the sweep pattern after a sufficiently long time, as the sweeping pattern is fixed. In this section, we consider the case where the jammer adopts an unpredictable jamming strategy such as the reactive sweep jamming or the $G$-memory jamming for $G\in[0:T-2]$. We note that the $G$-memory jamming is more unpredictable for smaller $G$ since the jammer selects $m$ target channels uniformly at random out of $M-Gm$ channels. 

We run numerous simulations with various parameter settings to compare the EDSRs for $G$-memory jamming and the reactive sweep jamming strategies when the S-R pairs apply corresponding optimal frequency hopping policies, and Fig.~\ref{Fig13} obtained from a certain parameter setting shows a typical tendency.\footnote{Note that the basic sweep jamming corresponds to the $(T-1)$-memory jamming.} Main observations from the simulation results are summarized in the following: 
\begin{itemize}
\item The EDSR of the reactive sweep jamming is the same with that of the basic sweep jamming for $n=1$ as proved in Theorem \ref{thm2}. But it becomes larger for $n\geq 2$, since the reactive sweep jammer can newly start sweeping even if an S-R pair stays in the same channel and hence the S-R pair has a lower risk of being jammed than the basic sweep jamming environment.
\item The EDSR of the $G$-memory jamming for $G\in[0:T-1]$ is non-increasing in $G$ since the probability that an S-R pair staying in a same channel is jammed increases in $G$. Let $K^*$ denote the optimal staying-threshold against the basic sweep jamming. Interestingly, we could check numerically that the EDSR is strictly decreasing for $G\in[0:K^*+1]$ and remains the same for $G\in[K^*+1:T-1]$ for any $n\geq 1$. For the setting of Fig. \ref{Fig13}, the optimal staying-threshold $K^*$ is given as $5$ for $n\in[1:6]$ and $6$ for  $n\in[7:10]$. We can see that the EDSR of the $6$-memory jamming is the same with that of the basic sweep jamming for $n\in[1:6]$ and there exists a gap for $n\in [7:10]$. For the $7$-memory jamming, its EDSR is the same with the basic sweep jamming for all $n\in [1:10]$. 
\end{itemize} 

\begin{figure}
\centering
\includegraphics[width=0.8\columnwidth]{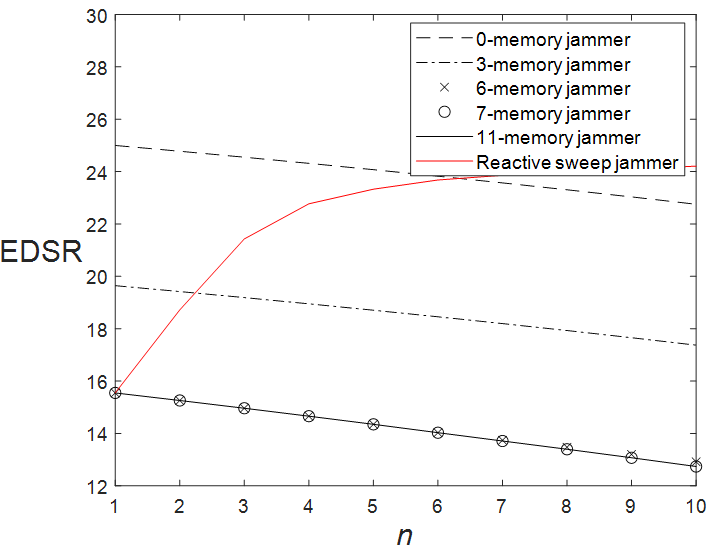}
\caption{The EDSRs according to the number $n$ of users for the reactive sweep jamming and the ($0,3,6,7,11$)-memory jamming strategies when the users apply corresponding optimal frequency hopping policies. The all-hopping collision avoidance protocol and parameters $R=5$, $C=5$, $L=20$, $M=60$, $m=5$, and $\gamma=0.9$ are applied, which implies that $T=12$ and hence $11$-memory jamming corresponds to the basic sweep jamming. In this setting, the optimal staying-threshold $K^*$ is given as $5$ for $n\in[1:6]$ and $6$ for $n\in[7:10]$.}\label{Fig13}
\end{figure}

The aforementioned analysis implies that if the optimal staying-threshold $K^*$ lies in $[0:T-3]$, it is preferable for the jammer to apply the $(K^*+1)$-memory jamming for any $n\geq 1$, since it achieves the same EDSR as the basic sweep jamming, which is the best jamming strategy against the $K^*$-staying policy, and also it is unpredictable. Furthermore, when $K^*\in [0:T-3]$, various numerical results show that an optimal frequency hopping policy against the $(K^*+1)$-memory jamming is the $K^*$-staying policy. This implies that for $K^*\in [0:T-3]$, $(K^*+1)$-memory jamming and the $K^*$-staying policy establish an equilibrium. This analysis suggests that if the jammer knows that $K^*$ lies in $[0:T-3]$, it is preferable for the jammer to apply the $(K^*+1)$-memory jamming, but in general it is hard for the jammer to know $K^*$ as it is a function of user parameters such as $R, C, L$ etc.

\begin{remark}\label{rmk5}
Against the $(K^*+1)$-memory jamming, we can analytically show that an optimal frequency hopping policy is a $\tilde{K}$-staying policy where $\tilde{K}\in [0:K^*,\infty]$ and the optimal staying-threshold $\tilde{K}$ can be obtained by finite iterative updates, which can be proved in a similar manner as in Theorems~\ref{thm3} and \ref{thm4}. Then, by extending the proof of Theorem~\ref{thm5}, we can show that if $\tilde{K}\in[0:K^*]$, the $(K^*+1)$-memory jamming and the $\tilde{K}$-staying policy establish an equilibrium. Proving exact $\tilde{K}$ value is a challenging problem since Bellman optimality equation does not have a closed form solution, but various simulations consistently show that $\tilde{K}$ equals $K^*$. 
\end{remark}

\section{Conclusion}\label{sec7} 
It was shown that the $K^*$-staying frequency hopping policy with optimal staying-threshold $K^*$ and the basic sweep jamming strategy establish an equilibrium of the anti-jamming game in multi-band wireless ad hoc networks. Various numerical results were provided to show the effects of the reward parameters and collision avoidance protocols on the optimal staying-threshold $K^*$ and the expected rewards at the equilibrium. Furthermore, we discussed about the equilibrium when the jammer adopts some unpredictable jamming strategies.  

Let us conclude with a final remark. We assumed that the users communicate with a fixed rate, but the users can increase the EDSR at an equilibrium by changing their communication rates adaptively. The arms race with rate adaptation was studied for the single-user environment \cite{Hanawal:2016}. It would an interesting further work to study the effect of the rate adaptation in the anti-jamming game in multi-user environments. 
\bibliographystyle{IEEEtran}
\bibliography{ref2}

\end{document}